\documentclass[conference]{IEEEtran}
\usepackage{cite}
\usepackage{graphicx,epsfig,amssymb}
\usepackage[T1]{fontenc}
\usepackage{amsthm,amsfonts}
\usepackage[cmex10]{amsmath}
\usepackage{amssymb}

\usepackage{array}
\begin{document}
%
\title{Coverage Analysis of Heterogeneous Cellular Networks in Urban Areas}


%
\author{\IEEEauthorblockN{Bin Yang\IEEEauthorrefmark{1},
Guoqiang Mao\IEEEauthorrefmark{2}\IEEEauthorrefmark{5}\IEEEauthorrefmark{1}\IEEEauthorrefmark{4},
Xiaohu Ge\IEEEauthorrefmark{1},
Hsiao-Hwa Chen\IEEEauthorrefmark{3},
Tao Han\IEEEauthorrefmark{1},
Xuefei Zhang\IEEEauthorrefmark{4}
}
\IEEEauthorblockA{\IEEEauthorrefmark{1} School of Electronic Information \& Communications,
Huazhong University of Science \& Technology, Wuhan, China
}
\IEEEauthorblockA{\IEEEauthorrefmark{2} School of Computing and Communication,
University of Technology Sydney, Australia}
\IEEEauthorblockA{\IEEEauthorrefmark{5} Data61 Australia}
\IEEEauthorblockA{\IEEEauthorrefmark{3} Department of Engineering Science, National Cheng Kung University, Taiwan}
\IEEEauthorblockA{\IEEEauthorrefmark{4} School of Information and Communication Engineering, Beijing University of Posts and Telecommunications, Beijing, China\\
Corresponding Author: Xiaohu Ge, Email: xhge@mail.hust.edu.cn}
}


\maketitle

\begin{abstract}
In this article, a network model incorporating both line-of-sight (LOS) and non-line-of-sight (NLOS) transmissions is proposed to investigate impacts of blockages in urban areas on heterogeneous network coverage performance. Results show that co-existence of NLOS and LOS transmissions has a significant impact on network performance. We find in urban areas, that deploying more BSs in different tiers is better than merely deploying all BSs in the same tier in terms of coverage probability.
\end{abstract}
\vspace{2 ex}



%

\section{Introduction}
With the tremendous growth of mobile subscribers and mobile connected devices, mobile broadband traffic has exhibited unprecedented growth. According to the forecast by Cisco \cite{Cisco15}, there will be 11.5 billion mobile-connected devices by 2019, which suggests that globally mobile data traffic is expected to grow to 24.3 exabytes (EB) per month by 2019 -- nearly a tenfold increase over that in 2014. The ever increasing mobile data traffic propels us to seek new techniques to handle the challenge. Deploying complementary small-cell networks like femtocells and picocells in the place where the need for traffic is "hot" emerges as a promising solution in the 5G wireless networks \cite{Ge5G14, Ge5G16, YangA15,Ge15}.

In urban areas, radio propagation is more complicated than rural areas due to high-rise buildings, trees, etc. Through reflection, diffraction, and even blockage, buildings not only attenuate the received signal power but also weaken the undesired signal power, i.e., the interference. Consequently it often occurs that the strongest signal does not come from the geographically nearest BS which renders location-based cell association scheme \cite{Mukherjee12Distribution} ineffective.

In \cite{Choi14On}, Choi analyzed blockage effects of a millimeter-wave cellular system. Different from traditional cellular signals, a mobile user (MU) can only communicate with BSs with line-of-sight (LOS) connection in millimeter wave systems. Therefore, outage probability obtained assuming millimeter wave communication will be higher than that expected in real cellular communications. Bai \textit{et al.} \cite{Bai14Analysis} analyzed the performance impact of large-scale blockage effects, which not only capture 2D shape of buildings, but also the height of buildings, in a microcell network. In their work a MU can only connect to its nearest LOS BS, which of course does not reflect the reality when a MU is in a central business district (CBD) with high-rise buildings or is located in an office building.

In this paper, we develop a model to analyze the coverage performance of heterogeneous networks incorporating both LOS and NLOS connections, which is considered typical for urban environment. Different from related work in which a BS is connected to the geographically nearest BS, which does not reflect the reality in urban environment, in our work, we consider that a MU is associated with a BS that delivers the strongest received signal-to-interference-plus-noise ratio (SINR) where both LOS and NLOS connections are considered in the analysis. The main contributions of this paper are summarized as follows.
\begin{enumerate}
\item The coverage performance in a heterogeneous network is analyzed considering both NLOS and LOS transmissions and that a MU is associated with the BS that delivers the strongest SINR.
\item An analytical expression for the coverage probability is derived. This is distinct from previous work where the theoretical analysis is conducted assuming that a MU is associated with its nearest LOS BS for analytical tractability.
\item Through results, We find in urban areas, that deploying more BSs in different tiers is better than merely deploying all BSs in the same tier in terms of coverage probability.
\end{enumerate}

The remainder of this paper is organized as follows. Section~\ref{sec2} describes the system model. While in Section~\ref{sec3}, the per tier coverage probability and the coverage probability of the entire network are derived. Results and performance analysis are given in Section~\ref{sec4}. Finally, Section~\ref{sec5} concludes this paper.

\section{System Model}
\label{sec2}
We consider a $K$-tier heterogeneous cellular network which
consists of macrocells, picocells, femtocells, etc and focus on the analysis in downlink coverage performance. The BSs of each tier are assumed to be spatially distributed on an infinite plane following
independent homogeneous Poisson point processes (PPPs) denoted by $\Phi_{k}$, $k\in\left\{ 1,2,\cdots,K\right\} $ with intensities $\lambda_{k}$, $k\in\left\{ 1,2,\cdots,K\right\} $.
MUs are located according to a homogeneous PPP denoted by $\Phi_{u}$ with intensity $\lambda_{u}$. BSs of the same tier transmit using the same power $P_{t}^{\left(k\right)}$ and share the same bandwidth $B_{k}$. BSs belonging to different tiers use different power
and orthogonal bandwidth for transmission. Therefore there is no inter-tier interference. Furthermore, within a cell, MUs use different frequency bandwidth for downlink and uplink transmission and therefore there is no intra-cell interference for downlink transmission analysis in our paper. However, BSs of the same tier may interfere each other and generate the inter-cell interference which is the main focus of this paper. The assumption that both BSs and MUs are homogeneously distributed over space makes our analysis tractable with a minor loss of accuracy \cite{Andrew11A}.

\subsection{Cell Association Scheme}
\label{2-A}
Taking both NLOS and LOS transmission into consideration, the nearest BS of a tier may not be the best candidate BS in that tier to associate with. More specifically, the cell association decision can be divided into the following two major steps.
\begin{enumerate}
\item If a MU requests to connect to a BS, it will firstly choose $n$ nearest BSs from each tier to form the set of candidate BSs. The candidate BS set is denoted by $\Omega_{B}=\left\{ \left(k,i\right)|k\in\left\{ 1,2,\cdots,K\right\} ,i\in\left[1,n\right]\right\} $, where $\left(k,i\right)$ represents the $i$-th nearest BS in the $k$-th tier. For example, $\left(1,3\right)$ means the 3rd nearest BS in the 1st tier. The total number of candidate BSs in $\Omega_{B}$ is $nK$.
\item The MU will then choose which BS to associate with from $\Omega_{B}$ based on its received SINR. The BS in $\Omega_B$ that offers the highest received SINR to the MU will be chosen.
\end{enumerate}

Apparently, when the value of $n$ is chosen to be sufficiently large, the coverage performance of the cell association scheme described above will resemble that of the cell association scheme that chooses the BS with the highest SINR at the MU from all tiers.

\subsection{Signal Propagation Model}
\label{2-B}

We consider both NLOS and LOS transmissions. The occurrence of NLOS or LOS transmissions depend on various environmental factors, including geographical structure, distance and clusters. N. Blaunstein \cite{Blaunstein98Parametric} gave the probability for NLOS transmissions
\begin{equation}
p_{N}\left(R\right)=1-e^{-\kappa R}.\label{eq:p_NLOS}
\end{equation}
And the probability for LOS transmissions is given by
\begin{equation}
p_{L}\left(R\right)=1-p_{N}\left(R\right)=e^{-\kappa R},\label{eq:p_LOS}
\end{equation}
where $\kappa$ is a parameter determined by the intensity and the mean length of the blockage lying in the visual path between a MU and a BS, and $R$ denotes the distance between the MU and the BS. Bai \cite{Bai14Analysis} extended N. Blaunstein's work by using random shape theory which considers that $\kappa$ is not only determined by the mean length but also the mean width.

In general, NLOS and LOS transmissions will experience different path losses, which is captured by the following model
\[
PL|NLOS\left[\textrm{dB}\right]=A_{N}^{\left(k\right)}+\alpha_{N}^{\left(k\right)}10\lg R+\xi_{N}^{\left(k\right)},
\]
\[
PL|LOS\left[\textrm{dB}\right]=A_{L}^{\left(k\right)}+\alpha_{L}^{\left(k\right)}10\lg R+\xi_{L}^{\left(k\right)},
\]
where $|LOS (NLOS)$ means the LOS (NLOS) transmission, $A_{N}^{\left(k\right)}$
and $A_{L}^{\left(k\right)}$ are constants determined by the transmission frequency, $\alpha_{N}^{\left(k\right)}$ and $\alpha_{L}^{\left(k\right)}$
are respective path loss exponents, $\xi_{N}^{\left(k\right)}$ and $\xi_{L}^{\left(k\right)}$
are independent Gaussian random variables with zero means, i.e., $\xi_{N}^{\left(k\right)}\sim\mathcal{N}\left(0,\left(\sigma_{N}^{\left(k\right)}\right)^{2}\right)$
and $\xi_{L}^{\left(k\right)}\sim\mathcal{N}\left(0,\left(\sigma_{L}^{\left(k\right)}\right)^{2}\right)$, reflecting the attenuation caused by flat fading. The corresponding parameters in this model can be found in \cite{Anderson14Fixed}. Accordingly, the received signal power without and with LOS transmission link in
W (watt) are
\[
P_{i}^{\left(k\right)}|NLOS=B_{N}^{\left(k\right)}\left(R_{i}^{\left(k\right)}\right)^{-\alpha_{N}^{\left(k\right)}}\exp\left(\beta\xi_{N}^{\left(k\right)}\right),
\]
and
\[
P_{i}^{\left(k\right)}|LOS=B_{L}^{\left(k\right)}\left(R_{i}^{\left(k\right)}\right)^{-\alpha_{L}^{\left(k\right)}}\exp\left(\beta\xi_{L}^{\left(k\right)}\right),
\]
respectively, where $R_{i}^{\left(k\right)}$ denotes the distance between a MU and the BS $\left(k,i\right)$,   $B_{N}^{\left(k\right)}=P_{t}^{\left(k\right)}\cdot10^{-{A_{N}^{\left(k\right)}}/{10}}$, $B_{L}^{\left(k\right)}=P_{t}^{\left(k\right)}\cdot10^{-{A_{L}^{\left(k\right)}}/{10}}$ and $\beta=-{\ln10}/{10}$ are constants.

Therefore, the received power of the typical MU located at the origin $o$ from the BS $\left(k,i\right)$ is given by
\begin{align}
P_{i}^{\left(k\right)}\left(R_{i}^{\left(k\right)}\right) & =\mathbb{I}_{i}^{\left(k\right)}P_{i}^{\left(k\right)}|\textrm{NLOS}+\left(1-\mathbb{I}_{i}^{\left(k\right)}\right)P_{i}^{\left(k\right)}|\textrm{LOS},\label{eq:rec_power}
\end{align}
where $\mathbb{I}_{i}^{\left(k\right)}$ is a random indicator variable equal to 1 for NLOS transmission and 0 otherwise, and the corresponding
probabilities are $p_{N}\left(R\right)$ and $p_{L}\left(R\right)$ given by (\ref{eq:p_NLOS}) and (\ref{eq:p_LOS}), respectively.

We consider an interference-limited system and the impact of noise is thus ignored \cite{Andrew11A}. For downlink transmission, the signal-to-interference ratio (SIR) experienced by the typical MU associated with the BS $(k,i)$ is expressed as follows
\begin{align}
\textrm{SIR}_{i}^{\left(k\right)} & =\frac{S}{I}=\frac{P_{i}^{\left(k\right)}}{\underset{(k,j)\in\Omega_{k}^{'}}{\sum}P_{j}^{\left(k\right)}},
\end{align}
where $\Omega_{k}^{'}$ is the Palm point process \cite{Stoyan96Stochastic} representing the set of interferers in the $k$-th tier network.

\section{The Coverage Probability}
\label{sec3}
A MU is considered to be covered if the SIR from any of the BS in the set of candidate BSs, formed by choosing $n$ nearest BSs from each tier, is greater than or equal to a prescribed threshold $\gamma$. In the following derivation, a conditional coverage probability will be firstly obtained with Laplace transform and then the unconditional probability will be derived by integrating with respect to random variables.

\subsection{An Analysis of the Conditional Coverage Probability}
\label{3-A}

\newtheorem{lemma}{Lemma}
\begin{lemma}
  For non-negative set $\Xi=\left\{ a_{q}\right\} $, $q\in\mathbb{N}$, if and only if $a_{m}>a_{n}$, then $\frac{a_{m}}{\underset{q\neq m}{\sum}a_{q}+W}>\frac{a_{n}}{\underset{q\neq n}{\sum}a_{q}+W}$, $\forall a_{m},a_{n}\in\Xi$.
\end {lemma}
\begin{proof}
   For non-negative set $\Xi=\left\{ a_{q}\right\} $, $q\in\mathbb{N}$, if and only if $a_{m}>a_{n}$, $\frac{a_{m}}{\underset{q}{\sum}a_{q}+W}>\frac{a_{n}}{\underset{q}{\sum}a_{q}+W}$, thus $\frac{a_{m}}{\underset{q}{\sum}a_{q}+W-a_{m}}>\frac{a_{n}}{\underset{q}{\sum}a_{q}+W-a_{n}}$, which completes the proof.
\end{proof}
Let $\mathsf{M}_{m}^{\left(k\right)}$ denote the event $\left\{ \arg\underset{i\in\left[1,n\right]}{\max}\textrm{SIR}_{i}^{\left(k\right)}=m\right\} $,
meaning that among the $n$ nearest BSs in the $k$-th tier, the maximum SIR comes from the $m$-th nearest BS in that tier. Accordingly, the typical MU will
connect to BS $\left(k,m\right)$ if we restrict that only BSs in the $k$-th tier are available. Conditioned on event $\mathsf{M}_{m}^{\left(k\right)}$, power received from the $m$-th nearest BS in the $k$-th tier $P_{m}^{\left(k\right)} = t$ and distances $\left\{ R_{i}^{\left(k\right)}=r_{i}^{\left(k\right)}\text{, i\ensuremath{\in\left[1,n\right]} }\right\} $, we obtain the conditional coverage probability that the typical MU is covered by the BS $(k,m)$ as follows

\begin{align}
 & \quad\Pr\left(\textrm{SIR}_{m}^{\left(k\right)}>\gamma\left|\mathsf{M}_{m}^{\left(k\right)},P_{m}^{\left(k\right)}=t,\left\{ R_{i}^{\left(k\right)}\right\} \right.\right)\nonumber \\
 & =\Pr\left(\left.\frac{P_{m}^{\left(k\right)}}{I_{m}^{\left(k\right)}}>\gamma\right|\mathsf{M}_{m}^{\left(k\right)},P_{m}^{\left(k\right)}=t,\left\{ R_{i}^{\left(k\right)}\right\} \right)\nonumber \\
 & =F_{I_{m}^{\left(k\right)}}\left(\left.\frac{P_{m}^{\left(k\right)}}{\gamma}\right|\mathsf{M}_{m}^{\left(k\right)},P_{m}^{\left(k\right)}=t,\left\{ R_{i}^{\left(k\right)}\right\} \right)\nonumber \\
 & \stackrel{\left(a\right)}{=}F_{I_{m}^{\left(k\right)}}\left(\left.\frac{t}{\gamma}\right|\underset{j\in\left[1,n\right],i\neq m}{\bigcap}P_{j}^{\left(k\right)}\leq t,P_{m}^{\left(k\right)}=t,\left\{ R_{i}^{\left(k\right)}=r_{i}^{\left(k\right)}\right\} \right),
\end{align}
with
\begin{equation}
I_{m}^{\left(k\right)}=I_{m,\leq n}^{\left(k\right)}+I_{m,>n}^{\left(k\right)}=\sum_{j=1,j\neq m}^{n}P_{j}^{\left(k\right)}+\sum_{j=n+1}^{\infty}P_{j}^{\left(k\right)},
\end{equation}
where the last step $\left(a\right)$ is obtained followed by Lemma 1 that event $\left\{ \arg\underset{i\in\left[1,n\right]}{\max}\textrm{SIR}_{i}^{\left(k\right)}=m\right\} $
is equivalent to event $\left\{\arg\underset{i\in\left[1,n\right]}{\max}P_{i}^{\left(k\right)}=m\right\}$,
$\gamma$ is the SIR threshold, $I_{m}^{\left(k\right)}$ is the aggregate
interference power experienced by the typical MU connecting to the BS $\left(k,m\right)$,
$I_{m,\leq n}^{\left(k\right)}$ is the aggregate interference signal
power from the $n-1$ candidate BSs excluding the m-th nearest BS in the $k$-th tier, $I_{m,>n}^{\left(k\right)}$
is the aggregate interference power from BSs located outside the disk area with radius $R_{n}^{\left(k\right)}$ in the $k$-th
tier, e.g. the $n+1$-th nearest BS, the $n+2$-th nearest BS, ... and so on, and $F_{I_{m}^{\left(k\right)}}\left(x\left|\mathsf{M},\left\{ R_{i}^{\left(k\right)}=r_{i}^{\left(k\right)}\right\} \right.\right)$
is the conditional cumulative distribution function (CDF) of the aggregate
interference $I_{m}^{\left(k\right)}$.

To obtain the CDF of $I_{m}^{\left(k\right)}$, we firstly derive
the Laplace transform (LT) of $I_{m,\leq n}^{\left(k\right)}$ and
$I_{m,>n}^{\left(k\right)}$, respectively.

\subsubsection{The LT of $I_{m,\leq n}^{\left(k\right)}$}

Conditioned on $\left\{ R_{i}^{\left(k\right)}=r_{i}^{\left(k\right)}\text{, i\ensuremath{\in\left[1,n\right]} }\right\} $, the distribution of received power $P_{j}^{\left(k\right)}$,
$j\neq m$, is derived as follows
\begin{align}
 & \quad F_{P_{j}^{\left(k\right)}}\left(x\left|R_{j}^{\left(k\right)}=r_{j}^{\left(k\right)}\right.\right)=\Pr\left(\left.P_{j}^{\left(k\right)}\leq x\right|R_{j}^{\left(k\right)}=r_{j}^{\left(k\right)}\right)\nonumber \\
 & \stackrel{\left(a\right)}{=}\Pr\left(\left.P_{j}^{\left(k\right)}\leq x\right|NLOS,R_{j}^{\left(k\right)}=r_{j}^{\left(k\right)}\right)\cdot p_{N}\left(r_{j}^{\left(k\right)}\right)\nonumber \\
 & \quad+\Pr\left(\left.P_{j}^{\left(k\right)}\leq x\right|LOS,R_{j}^{\left(k\right)}=r_{j}^{\left(k\right)}\right)\cdot p_{L}\left(r_{j}^{\left(k\right)}\right)\nonumber \\
 & =\left(1-e^{-\kappa r_{j}^{\left(k\right)}}\right)\int_{0}^{x}f_{P_{j}^{\left(k\right)}}\left(z|NLOS,R_{j}^{\left(k\right)}=r_{j}^{\left(k\right)}\right)\textrm{d}z+\nonumber \\
 & \quad e^{-\kappa r_{j}^{\left(k\right)}}\int_{0}^{x}f_{P_{j}^{\left(k\right)}}\left(z|LOS,R_{j}^{\left(k\right)}=r_{j}^{\left(k\right)}\right)\textrm{d}z,\label{eq:F_P_1}
\end{align}
where $\left(a\right)$ follows from the law of total probability,  $f_{P_{j}^{\left(k\right)}}\left(z|NLOS,R_{j}^{\left(k\right)}=r_{j}^{\left(k\right)}\right)$
and $f_{P_{j}^{\left(k\right)}}\left(z|LOS,R_{j}^{\left(k\right)}=r_{j}^{\left(k\right)}\right)$
are the PDF of received signal power $P_{j}^{\left(k\right)}$ conditioned
on NLOS and LOS transmissions, respectively.

The PDF of received signal power $P_{j}^{\left(k\right)}$ conditioning
on NLOS transmission and the distance $R_{j}^{\left(k\right)}=r_{j}^{\left(k\right)}$, i.e., $f_{P_{j}^{\left(k\right)}}\left(z|NLOS,R_{j}^{\left(k\right)}=r_{j}^{\left(k\right)}\right)$
is given by
\begin{align}
 & \quad f_{P_{j}^{\left(k\right)}}\left(z|NLOS,R_{j}^{\left(k\right)}=r_{j}^{\left(k\right)}\right)\nonumber \\
 & \stackrel{\left(a\right)}{=}\left|\frac{\textrm{d}}{\textrm{d}z}\left[P_{j}^{\left(k\right)^{*}}\left(x\right)\right]\right|\cdot f_{\xi_{N}^{\left(k\right)}}\left[P_{j}^{\left(k\right)^{*}}\left(z\right)\right]\nonumber \\
 & =\frac{1}{z\sigma_{sN}^{\left(k\right)}\sqrt{2\pi}}\exp\left[-\left(\ln z-\mu_{sN}^{\left(k\right)}\right)^{2}\left/2\left(\sigma_{sN}^{\left(k\right)}\right)^{2}\right.\right],\label{eq:P_NLOS,r}
\end{align}
where $\left(a\right)$ is obtained by applying the change-of-variables
rule on the density function of a normal distribution, $P_{j}^{\left(k\right)^{*}}\left(z\right)$
denotes the inverse function of $P_{i}^{\left(k\right)}|NLOS$, thus $\xi_{N}^{\left(k\right)}=P_{j}^{\left(k\right)^{*}}\left(P_{j}^{\left(k\right)}\right)=\frac{1}{\beta}\left(\ln P_{j}^{\left(k\right)}-\mu_{sN}^{\left(k\right)}\right)$,
$\mu_{sN}^{\left(k\right)}=\ln B_{N}^{\left(k\right)}-\alpha_{N}^{\left(k\right)}\ln r_{j}^{\left(k\right)}$
and $\left(\sigma_{sN}^{\left(k\right)}\right)^{2}=\left(\beta\sigma_{N}^{\left(k\right)}\right)^{2}$.
Obviously, $P_{j}^{\left(k\right)}$ is log-normal distributed conditioned
on NLOS transmission and the distance $R_{j}^{\left(k\right)}=r_{j}^{\left(k\right)}$, i.e., $P_{j}^{\left(k\right)}\sim\ln\mathcal{N}\left(\mu_{sN}^{\left(k\right)},\sigma_{sN}^{\left(k\right)}\right)$.

Similarly, the PDF of received signal power $P_{j}^{\left(k\right)}$
conditioning on LOS transmission and the distance $R_{j}^{\left(k\right)}=r_{j}^{\left(k\right)}$, i.e., $f_{P_{j}^{\left(k\right)}}\left(z|LOS,R_{j}^{\left(k\right)}=r_{j}^{\left(k\right)}\right)$
is given by
\begin{align}
 & \quad f_{P_{j}^{\left(k\right)}}\left(z|LOS,R_{j}^{\left(k\right)}=r_{j}^{\left(k\right)}\right)\nonumber \\
 & =\frac{1}{z\sigma_{sL}^{\left(k\right)}\sqrt{2\pi}}\exp\left[-\left(\ln z-\mu_{sL}^{\left(k\right)}\right)^{2}\left/2\left(\sigma_{sL}^{\left(k\right)}\right)^{2}\right.\right],\label{eq:f_P_LOS,r}
\end{align}
where $\mu_{sL}^{\left(k\right)}=\ln B_{L}^{\left(k\right)}-\alpha_{L}^{\left(k\right)}\ln r_{j}^{\left(k\right)}$
and $\left(\sigma_{sL}^{\left(k\right)}\right)^{2}=\left(\beta\sigma_{L}^{\left(k\right)}\right)^{2}$.
Plugging (\ref{eq:P_NLOS,r}) and (\ref{eq:f_P_LOS,r}) into (\ref{eq:F_P_1}) , the distribution of received power $P_{j}^{\left(k\right)}$,
$j\neq m$, i.e., $F_{P_{j}^{\left(k\right)}}\left(x\left|R_{j}^{\left(k\right)}=r_{j}^{\left(k\right)}\right.\right)$, can be obtained:
\begin{align}
 & \quad F_{P_{j}^{\left(k\right)}}\left(x\left|R_{j}^{\left(k\right)}=r_{j}^{\left(k\right)}\right.\right)=\frac{1}{2}\Bigg[1+\left(1-e^{-\kappa r_{j}^{\left(k\right)}}\right)\nonumber \\
 & \textrm{erf}\left(\frac{\ln x-\mu_{sN}^{\left(k\right)}}{\sqrt{2}\sigma_{sN}^{\left(k\right)}}\right)+e^{-\kappa r_{j}^{\left(k\right)}}\textrm{erf}\left(\frac{\ln x-\mu_{sL}^{\left(k\right)}}{\sqrt{2}\sigma_{sL}^{\left(k\right)}}\right)\Bigg],\label{eq:F_P_2}
\end{align}
where $\textrm{erf}\left(\cdot\right)$ is the error function.

The PDF of the received power given the distance $R_{j}^{\left(k\right)}=r_{j}^{\left(k\right)}$
is obtained by taking the derivative of $F_{P_{j}^{\left(k\right)}}\left(x\left|R_{j}^{\left(k\right)}=r_{j}^{\left(k\right)}\right.\right)$
with respect to $x$:
\begin{align}
 & f_{P_{j}^{\left(k\right)}}\left(x\left|R_{j}^{\left(k\right)}=r_{j}^{\left(k\right)}\right.\right)=\frac{e^{-\kappa r_{j}^{\left(k\right)}}}{x\sigma_{sL}^{\left(k\right)}\sqrt{2\pi}}\exp\left[\frac{-\left(\ln x-\mu_{sL}^{\left(k\right)}\right)^{2}}{2\left(\sigma_{sL}^{\left(k\right)}\right)^{2}}\right]\nonumber \\
 & +\frac{\left(1-e^{-\kappa r_{j}^{\left(k\right)}}\right)}{x\sigma_{sN}^{\left(k\right)}\sqrt{2\pi}}\exp\left[\frac{-\left(\ln x-\mu_{sN}^{\left(k\right)}\right)^{2}}{2\left(\sigma_{sN}^{\left(k\right)}\right)^{2}}\right].
\end{align}

Conditioned on $\underset{j\in\left[1,n\right],j\neq m}{\bigcap}P_{j}^{\left(k\right)}\leq t$,
$P_{m}^{\left(k\right)}=t,$ and $\left\{ R_{i}^{\left(k\right)}=r_{i}^{\left(k\right)}\text{, i\ensuremath{\in\left[1,n\right]} }\right\} $, the PDF of $P_{j}^{\left(k\right)}$ is then given
by
\begin{align}
 & \quad f_{P_{j}^{\left(k\right)}}\left(x\left|P_{j}^{\left(k\right)}\leq t,P_{m}^{\left(k\right)}=t,R_{j}^{\left(k\right)}=r_{j}^{\left(k\right)}\right.\right)\nonumber \\
 & =\frac{f_{P_{j}^{\left(k\right)}}\left(x\left|R_{j}^{\left(k\right)}=r_{j}^{\left(k\right)}\right.\right)}{\int_{0}^{t}f_{P_{j}^{\left(k\right)}}\left(x\left|R_{j}^{\left(k\right)}=r_{j}^{\left(k\right)}\right.\right)\textrm{d}x}\nonumber \\
 & =2\Bigg\{\left(1-e^{-\kappa r_{j}^{\left(k\right)}}\right)\textrm{erfc}\left(-\frac{\ln t-\mu_{sN}^{\left(k\right)}}{\sqrt{2}\sigma_{sN}^{\left(k\right)}}\right)+e^{-\kappa r_{j}^{\left(k\right)}}\cdot\nonumber \\
 & \textrm{erfc}\left(-\frac{\ln t-\mu_{sL}^{\left(k\right)}}{\sqrt{2}\sigma_{sL}^{\left(k\right)}}\right)\Bigg\}^{-1}\Bigg\{\frac{e^{-\kappa r_{j}^{\left(k\right)}}}{x\sigma_{sL}^{\left(k\right)}\sqrt{2\pi}}\exp\left[\frac{-\left(\ln x-\mu_{sL}^{\left(k\right)}\right)^{2}}{2\left(\sigma_{sL}^{\left(k\right)}\right)^{2}}\right]\nonumber \\
 & +\frac{\left(1-e^{-\kappa r_{j}^{\left(k\right)}}\right)}{x\sigma_{sN}^{\left(k\right)}\sqrt{2\pi}}\exp\left[\frac{-\left(\ln x-\mu_{sN}^{\left(k\right)}\right)^{2}}{2\left(\sigma_{sN}^{\left(k\right)}\right)^{2}}\right]\Bigg\},0<x\leq t.
\end{align}
where $\textrm{erfc}\left(x\right)=1-\textrm{erf}\left(x\right)$
is the complementary error function.

Thus the LT of $P_{j}^{\left(k\right)}$ conditioned on $\underset{j\in\left[1,n\right],j\neq m}{\bigcap}P_{j}^{\left(k\right)}\leq t$,
$P_{m}^{\left(k\right)}=t,$ and $\left\{ R_{i}^{\left(k\right)}=r_{i}^{\left(k\right)}\text{, i\ensuremath{\in\left[1,n\right]} }\right\} $ is derived by using its definition

\begin{align}
 & \quad\mathcal{L}_{P_{j}^{\left(k\right)}}\left(s\right)\nonumber \\
 & =\int_{0}^{t}e^{-sx}f_{P_{j}^{\left(k\right)}}\left(x\left|P_{j}^{\left(k\right)}\leq t,P_{m}^{\left(k\right)}=t,R_{j}^{\left(k\right)}=r_{j}^{\left(k\right)}\right.\right)dx.\label{eq:L_P}
\end{align}

\begin{lemma}
  If conditioned on $\left\{ R_{i}^{\left(k\right)}=r_{i}^{\left(k\right)}\text{, i\ensuremath{\in\left[1,n\right]} }\right\} $, $P_{j}^{\left(k\right)}$, $j<n$ are independent of each other.
\end {lemma}
\begin{proof}
   From (\ref{eq:rec_power}), we find that if conditioned on $\left\{ R_{i}^{\left(k\right)}=r_{i}^{\left(k\right)}\right\} $, $P_{j}^{\left(k\right)}$ are determined by random variables $\mathbb{I}_{j}^{\left(k\right)}$, $\xi_{N}^{\left(k\right)}$ and $\xi_{L}^{\left(k\right)}$ which are all independent of each other. Thus $P_{j}^{\left(k\right)}$ are independent of each other, which completes the proof.
\end{proof}
By using Lemma 2, the moment generating function (MGF) of $I_{m,\leq n}^{\left(k\right)}={\sum_{j=1,j\neq m}^{n}}P_{j}^{\left(k\right)}$
is the product of the LTs of individuals
\begin{align}
 & \quad\mathcal{L}_{I_{m,\leq n}^{\left(k\right)}}\left(s\right)={\prod_{j=1,j\neq m}^{n}}\mathcal{L}_{P_{j}^{\left(k\right)}}\left(s\right)
\end{align}

\subsubsection{The LT of $I_{m,>n}^{\left(k\right)}$}

Next, the LT of $I_{m,>n}^{\left(k\right)}={\sum_{j=n+1}^{\infty}}P_{j}^{\left(k\right)}$ will be derived
conditioned on $\underset{j\in\left[1,n\right],j\neq m}{\bigcap}P_{j}^{\left(k\right)}\leq t$,
$P_{m}^{\left(k\right)}=t,$ and $\left\{ R_{i}^{\left(k\right)}=r_{i}^{\left(k\right)}\right\} $,
$i\in\left[1,n\right]$. Note that if conditioned on $\left\{ R_{i}^{\left(k\right)}=r_{i}^{\left(k\right)}\right\} $ only,
$P_{j}^{\left(k\right)}$, $j>n$ are independent of each other and also independent of $P_{j}^{\left(k\right)}$, $j\leq n$ as well because distance $R_{l}^{\left(k\right)}$, $l>n$ are not sequential by index as $R_{l}^{\left(k\right)}$,
$l\leq n$. For simplification, assume that the transmission of interference from BSs located outside the disk area with radius $R_{n}^{\left(k\right)}$ are all NLOS. This simplification can be justified by that BSs located far away are more likely to have NLOS paths. The LT of $I_{m,>n}^{\left(k\right)}$ is obtained as follows
\begin{align}
 & \quad\mathcal{L}_{I_{m,>n}^{\left(k\right)}}\left(s\right)=\mathbb{E}_{I_{m,>n}^{\left(k\right)}}\left(e^{-sI_{m,>n}^{\left(k\right)}}\right)\nonumber \\
 & =\mathbb{E}_{\Phi_{k},\xi_{N}^{\left(k\right)}}\left[\exp\left(-s\underset{l\in\Phi_{k}\left\backslash \Theta_{n}^{\left(k\right)}\right.}{\sum}B_{N}^{\left(k\right)}\left(r_{l}^{\left(k\right)}\right)^{-\alpha_{N}^{\left(k\right)}}e^{\beta\xi_{N}^{\left(k\right)}}\right)\right]\nonumber \\
 & =\mathbb{E}_{\Phi_{k}}\left\{ \underset{l\in\Phi_{k}\left\backslash \Theta_{n}^{\left(k\right)}\right.}{\prod}\mathbb{E}_{\xi_{N}^{\left(k\right)}}\exp\left[-sB_{N}^{\left(k\right)}\left(r_{l}^{\left(k\right)}\right)^{-\alpha_{N}^{\left(k\right)}}e^{\beta\xi_{N}^{\left(k\right)}}\right]\right\} \nonumber \\
 & \stackrel{\left(a\right)}{=}\exp\left\{ -2\pi\lambda_{k}\int_{r_{n}^{\left(k\right)}}^{\infty}\left[1-\varphi\left(v,s\right)\right]v\textrm{d}v\right\} ,\label{eq:LT_Im_gre}
\end{align}
with
\begin{align}
 & \varphi\left(v,s\right)=\nonumber \\
 & \int_{-\infty}^{\infty}\frac{1}{\sigma_{N}^{\left(k\right)}\sqrt{2\pi}}\exp\left[-sB_{N}^{\left(k\right)}v^{-\alpha_{N}^{\left(k\right)}}e^{\beta u}-\frac{u^{2}}{2\left(\sigma_{N}^{\left(k\right)}\right)^{2}}\right]\textrm{d}u.\label{eq:phi}
\end{align}
where $\Theta_{n}^{\left(k\right)}$ denotes the set of the locations of BS $\left(k,i\right)$,
$,i\in\left[1,n\right]$, $\left(a\right)$ follows from the probability
generating functional (PGFL) of the PPP \cite{Stoyan96Stochastic}.

At last, the PDF of $I_{m}^{\left(k\right)}$ is obtained by taking an
inverse LT of $\mathcal{L}_{I_{m}^{\left(k\right)}}\left(s\right)=\mathcal{L}_{I_{m,\leq n}^{\left(k\right)}}\left(s\right)\mathcal{L}_{I_{m,>n}^{\left(k\right)}}\left(s\right)$
, i.e.,
\begin{equation}
f_{I_{m}^{\left(k\right)}}\left(x\right)=\mathcal{L}_{I_{m}^{\left(k\right)}}^{-1}\left(s\right).
\end{equation}

Through derivations above, we get
\begin{align}
 & F_{I_{m}^{\left(k\right)}}\left(\left.\frac{t}{\gamma}\right|\underset{j\in\left[1,n\right],j\neq m}{\bigcap}P_{j}^{\left(k\right)}\leq t,P_{m}^{\left(k\right)}=t,\left\{ R_{i}^{\left(k\right)}=r_{i}^{\left(k\right)}\right\} \right)\nonumber \\
 & =\int_{0}^{\frac{t}{\gamma}}f_{I_{m}^{\left(k\right)}}\left(x\right)\textrm{d}x
\end{align}

\subsection{The Per Tier Coverage Probability and Coverage Probability}
\label{3-B}
The per tier coverage probability can be derived by de-conditioning with respect to $\underset{j\in\left[1,n\right],j\neq m}{\bigcap}P_{j}^{\left(k\right)}\leq t$,
$P_{m}^{\left(k\right)}=t$ and $\left\{ R_{i}^{\left(k\right)}=r_{i}^{\left(k\right)}\right\} $, respectively.

Firstly, we need to de-condition with respect to $\underset{j\in\left[1,n\right],j\neq m}{\bigcap}P_{j}^{\left(k\right)}\leq t$.
Noticing the conditional independence of $P_{j}^{\left(k\right)}, j\neq m$ (conditioned on their respective distances),  the probability of $\underset{j\in\left[1,n\right],j\neq m}{\bigcap}P_{j}^{\left(k\right)}\leq t$
conditioning on $P_{m}^{\left(k\right)}=t$ and $\left\{ R_{i}^{\left(k\right)}=r_{i}^{\left(k\right)}\text{, i\ensuremath{\in\left[1,n\right]} }\right\} $ is given by
\begin{align}
 & \quad\Pr\left(\left.\underset{j\in\left[1,n\right],j\neq m}{\bigcap}P_{j}^{\left(k\right)}\leq t\right|P_{m}^{\left(k\right)}=t,\left\{ R_{i}^{\left(k\right)}=r_{i}^{\left(k\right)}\right\} \right)\nonumber \\
 & ={\prod_{j=1,j\neq m}^{n}}\Pr\left(\left.P_{j}^{\left(k\right)}\leq t\right|R_{j}^{\left(k\right)}=r_{j}^{\left(k\right)}\right)\nonumber \\
 & ={\prod_{j=1,j\neq m}^{n}}F_{P_{j}^{\left(k\right)}}\left(t\left|R_{j}^{\left(k\right)}=r_{j}^{\left(k\right)}\right.\right).
\end{align}
De-conditioning with respect to $\underset{j\in\left[1,n\right],j\neq m}{\bigcap}P_{j}^{\left(k\right)}\leq t$,
we obtain
\begin{align}
 & F_{I_{m}^{\left(k\right)}}\left(\left.\frac{t}{\gamma}\right|P_{m}^{\left(k\right)}=t,\left\{ R_{i}^{\left(k\right)}=r_{i}^{\left(k\right)}\right\} \right)\nonumber \\
 & =\int_{0}^{\frac{t}{\gamma}}f_{I_{m}^{\left(k\right)}}\left(x\right)\textrm{d}x{\prod_{j=1,j\neq m}^{n}}F_{P_{j}^{\left(k\right)}}\left(t\left|R_{j}^{\left(k\right)}=r_{j}^{\left(k\right)}\right.\right).
\end{align}

Next, we obtain $F_{I_{m}^{\left(k\right)}}\left(\left.\frac{t}{\gamma}\right|\left\{ R_{i}^{\left(k\right)}=r_{i}^{\left(k\right)}\right\} \right)$
by de-conditioning with respect to $P_{m}^{\left(k\right)}=t$, which is derived as
follows
\begin{align}
 & F_{I_{m}^{\left(k\right)}}\left(\left.\frac{t}{\gamma}\right|\left\{ R_{i}^{\left(k\right)}=r_{i}^{\left(k\right)}\right\} \right)\nonumber \\
 & =\int_{0}^{\infty}\Big \{\int_{0}^{\frac{t}{\gamma}}f_{I_{m}^{\left(k\right)}}\left(x\right)\textrm{d}x{\prod_{j=1,j\neq m}^{n}}F_{P_{j}^{\left(k\right)}}\left(t\left|R_{j}^{\left(k\right)}=r_{j}^{\left(k\right)}\right.\right)\Big \}\cdot\nonumber \\
 & \quad f_{P_{m}^{\left(k\right)}}\left(t\left|R_{m}^{\left(k\right)}=r_{m}^{\left(k\right)}\right.\right)\textrm{d}t. \label{eq:FI_R}
\end{align}

As a final step, we shall de-condition with respect to $\left\{ R_{i}^{\left(k\right)}=r_{i}^{\left(k\right)}\right\} $
and the joint PDF of $\left\{ R_{i}^{\left(k\right)}=r_{i}^{\left(k\right)}\right\} $
is derived as follows.
\begin{lemma}
  The joint distance distribution up to $n$-th nearest neighbors in the $k$-th tier which is given by \cite{Moltchanov12Distance}
\begin{align}
 & \quad f_{k}\left(r_{1},r_{2},\ldots,r_{n}\right)\nonumber \\
 & =\begin{cases}
\left(2\pi\lambda_{k}\right)^{n}r_{1}r_{2}\ldots r_{n}e^{-\pi\lambda_{k}r_{n}^{2}}, & 0\leq r_{1}\leq r_{2}\leq\ldots\leq r_{n}\\
0, & otherwise
\end{cases} \label{eq:fR}
\end{align}
where $r_{i}$, $i\in\left[1,n\right]$ is  the $i$-th nearest distance
to the typical MU%
\footnote{We omit tier order $k$ in this subsections for notation simplification.%
}.
\end {lemma}

Combing equations (\ref{eq:FI_R}) and (\ref{eq:fR}), the unconditional probability $F_{I_{m}^{\left(k\right)}}\left(\frac{t}{\gamma}\right)$ can be obtained by de-conditioning with
respect to $\left\{ R_{i}^{\left(k\right)}=r_{i}^{\left(k\right)}\right\} $ as follows
\begin{align}
 & F_{I_{m}^{\left(k\right)}}\left(\frac{t}{\gamma}\right)\nonumber \\
 & =\underset{0\leq r_{1}^{\left(k\right)}\leq\ldots\leq r_{n}^{\left(k\right)}}{\idotsint}\int_{0}^{\infty}\Big \{\int_{0}^{\frac{t}{\gamma}}f_{I_{m}^{\left(k\right)}}\left(x\right)\textrm{d}x\nonumber \\
 & {\prod_{j=1,j\neq m}^{n}}F_{P_{j}^{\left(k\right)}}\left(t\left|R_{j}^{\left(k\right)}=r_{j}^{\left(k\right)}\right.\right)\Big \}\cdot f_{P_{m}^{\left(k\right)}}\left(t\left|R_{m}^{\left(k\right)}=r_{m}^{\left(k\right)}\right.\right)\textrm{d}t\cdot\nonumber \\
 & \left(2\pi\lambda_{k}\right)^{n}r_{1}^{\left(k\right)}r_{2}^{\left(k\right)}\ldots r_{n}^{\left(k\right)}e^{-\pi\lambda_{k}\left(r_{n}^{\left(k\right)}\right)^{2}}\textrm{d}r_{1}^{\left(k\right)}\textrm{d}r_{2}^{\left(k\right)}\ldots\textrm{d}r_{n}^{\left(k\right)}.
\end{align}
The per tier coverage probability is the summation of all possibilities
as follows
\begin{equation}
\boldsymbol{\textrm{P}_{c}^{\left(k\right)}}\left(\gamma\right)={\sum_{m=1}^{n}} F_{I_{m}^{\left(k\right)}}\left(\frac{t}{\gamma}\right).
\end{equation}
And the coverage probability of the whole tiers is given by
\begin{align}
 & \boldsymbol{\textrm{P}_{c}}\left(\gamma\right)=1-{\prod_{k=1}^{K}}\left[1-\boldsymbol{\textrm{P}_{c}^{\left(k\right)}}\left(\gamma\right)\right]\nonumber \\
 & =1-{\prod_{k=1}^{K}}\left[1-{\sum_{m=1}^{n}} F_{I_{m}^{\left(k\right)}}\left(\frac{t}{\gamma}\right)\right].
\end{align}


\section{Performance Analysis and Results}
\label{sec4}
This section presents results of previous sections, followed by discussions. \cite{Anderson14Fixed} suggest that usually $\alpha_{N}^{\left(k\right)}>\alpha_{L}^{\left(k\right)}$ and $\sigma_{N}^{\left(k\right)}>\sigma_{L}^{\left(k\right)}$ if we fix antenna types and heights. Let $\alpha_{N}^{\left(1\right)} = 4.28$, $\alpha_{L}^{\left(1\right)} = 2.42$, $\alpha_{N}^{\left(2\right)} = 3.75$ and $\alpha_{L}^{\left(2\right)} = 2.09$ in a 2-tier network. $A_{N}^{\left(1\right)}$, $A_{L}^{\left(1\right)}$, $A_{N}^{\left(2\right)}$ and $A_{L}^{\left(2\right)}$ are set to 2.7, 30.8, 32.9 and 41.4, respectively.

\begin{figure}
  \centering
  \includegraphics[width=7.2cm,height=6cm]{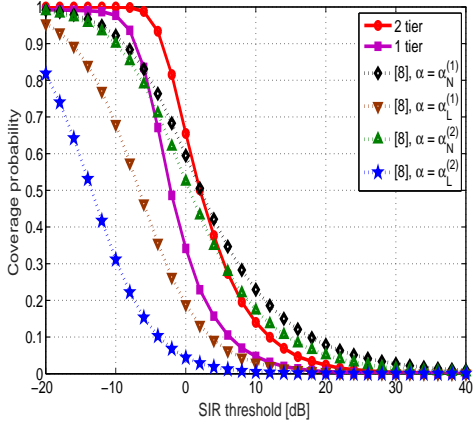}\\
  \caption{\small Coverage probability vs. SIR threshold in a 2-tier network ($K = 2$, $n = 5$, $P_{t}^{\left(1\right)} = 47\textrm{dBm}$, $P_{t}^{\left(2\right)} = 33\textrm{dBm}$).}\label{Fig2}
\end{figure}

\begin{figure}
  \centering
  \includegraphics[width=7.2cm,height=6cm]{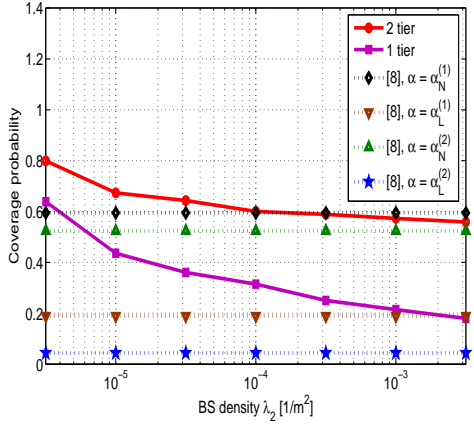}\\
  \caption{\small Coverage probability vs. 2nd tier BS density $\lambda_{2}$ in a 2-tier network ($K = 2$, $n = 2$, $P_{t}^{\left(1\right)} = 47\textrm{dBm}$, $P_{t}^{\left(2\right)} = 33\textrm{dBm}$).}\label{Fig3}
\end{figure}

Fig.~\ref{Fig2} shows the coverage probability with respect to the SIR threshold which varies from -20 dB to 40 dB. It is found that the coverage probability decreases with the increase of the SIR threshold, as the higher the SIR threshold, the more difficult for the received SIR at a MU to be higher than the SIR threshold. When the SIR threshold is fixed, 2-tier networks perform better than 1-tier networks. Besides, a comparison with \cite{Andrew11A} which does not consider NLOS and LOS transmissions is illustrated in the same figure. In \cite{Andrew11A}, the coverage probability is a monotonously increasing along with the path loss exponent $\alpha$. In our model, the coverage probability of both 2-tier and 1-tier networks have a similar trend with networks configured with NLOS path loss exponent, i.e., $\alpha_{N}^{\left(1\right)}$ or $\alpha_{N}^{\left(2\right)}$, which indicates that buildings and trees have a non-negligible impact on network performance.

The coverage probability vs. 2nd BS density $\lambda_{2}$ is given by Fig.~\ref{Fig3}. Through curves above, the coverage probability decreases with a slower and slower rate as 2nd BS density increases. While in \cite{Andrew11A}, the coverage probability is only a function of SIR threshold and path loss exponent if we ignore terminal noise.

Comparing Fig.~\ref{Fig2} with Fig.~\ref{Fig3}, we find that in urban areas \emph{dense BS deployment} do not always provide a better network performance. As Fig.~\ref{Fig2} shows, coverage probability in 2-tier networks with dense BSs is larger than that in 1-tier networks when the SIR threshold is fixed. While in Fig.~\ref{Fig3}, dense BSs  deployment weakens networks performance, which indicates that deploying more BSs in different tiers is better than deploying all BSs in the same tier in terms of coverage probability in urban areas. This is an effect caused by the co-existence of NLOS and LOS transmission in our model.
\section{Conclusions}
\label{sec5}
In this paper, we propose a heterogeneous network model considering LOS and NLOS transmission to study the coverage performance. The coverage probability is derived and analyzed with the assumption that both visible and invisible BSs are available for a MU, as long as the SIR threshold is satisfied. We also compare our work with \cite{Andrew11A} and obtain some interesting observations. As for our future work, channel
model shall be generalized and impacts of the number of
candidate BSs per tier should also be investigated.

\section*{Acknowledgment}

The authors would like to acknowledgement from the International Science and Technology Cooperation Program of China (Grant No. 2015DFG12580 and 2014DFA11640), the National Natural Science Foundation of China (NSFC) (Grant No. 61471180 and 61210002) and the Fundamental Research Funds for the Central Universities (HUST Grant No. 2015XJGH011 and 2015MS038). Guoqiang Mao's research is supported by Australian Research Council (ARC) Discovery projects DP110100538 and DP120102030 and NSFC (Grant No. 61428102).

\end{document}